\documentclass[10pt]{amsart}
 
\usepackage[utf8]{inputenc}
\usepackage[english]{babel}
 
\usepackage{url}
\usepackage{hyperref}
\usepackage{amsfonts,cancel}
\usepackage{amssymb}
\usepackage{amsthm}
\usepackage{amsopn}
\usepackage{amsmath,comment} 
\usepackage{hyperref}
\usepackage[dvipsnames]{xcolor}

\usepackage{graphicx}

\newtheorem{defi}{Definition}
\newtheorem{exm}{Example}
\newtheorem{exr}{Exercise}

\newtheorem{lem}{Lemma}
\newtheorem{thrm}{Theorem}

\newtheorem{rem}{Remark}

\newcommand{\Aut}{\mathop{\mathrm{Aut}}\nolimits} 
 
\newcommand{\Spec}{\mathop{\mathrm{Spec}}\nolimits}  

\allowdisplaybreaks



\begin{document}

	\title{Beta super-functions on super-Grassmannians 
	}
	\author[]{Mee Seong Im \and Micha\l{} Zakrzewski}
	\address{Department of Mathematical Sciences, United States Military Academy, West Point, NY 10996 USA}
	\email{meeseongim@gmail.com} 
	\address{Department of Mathematics, Jan Kochanowski University, Kielce, Poland}
	\email{zakrzewski@mimuw.edu.pl}
	
	\date{\today} 
	
	\maketitle
	
	\begin{abstract}
	Israel M. Gelfand gave a geometric interpretation for general hypergeometric functions as sections of the tautological bundle over a complex Grassmannian $G_{k,n}$. 
	In particular, the beta function can be understood in terms of $G_{2,3}$. In this manuscript, we construct one of the simplest generalizations of the Euler beta function by adding arbitrary-many odd variables to the classical setting. 
	We also relate the beta super-function to the gamma and the hypergeometric function. 
%
 
	\end{abstract} 

\maketitle 

\bibliographystyle{amsalpha}  
\setcounter{tocdepth}{0}

\section{Introduction}

 Bessel, Jacobi, Legendre, $3j$- and $6j$-symbols in quantum mechanics and many other classical special functions are special cases of hypergeometric functions. 
 In fact, many elementary and other important functions in mathematics and physics can be written in terms of hypergeometric functions, including the Euler beta and the gamma function. 
They can also be described as solitons of special differential equations, classified by singularities and exponents of certain differential equations. 
They have explicit integral and series representations, transformation and summation formulas, and other beautiful formulations relating various representations of hypergeometric functions. These functions are found and studied in combinatorics (cf. \cite{MR2343232, MR1429893}), Hilbert spaces as classical orthogonal polynomial bases (cf. \cite{koepf1996algorithms}), quantum physics in the form of harmonic analysis (cf. \cite{dattoli2010euler}), integrable systems of nonlinear differential equations as $q$-hypergeometric series using elliptic and theta-functions (cf. \cite{MR2506157, MR2044635}), and representation theory as matrix coefficients of Lie group representations (cf. \cite{MR1476496}).

The gamma function was constructed by Euler in an attempt to find an analytic continuation of the factorial function. It has a representation as an infinite integral and as a limit of a finite product, and it describes factorials in the special case with an integral domain. 
The beta function is one of the classical Euler integrals and it can define binomial coefficients after a certain adjustment of indices.  
It is a fundamental tool   
to systems of holonomic equations (\cite{MR902936}), as well as a fundamental special function in engineering (cf. \cite{MR3750264, MR3735044, MR3797763}), 
analysis (cf. \cite{MR1926172, MR951750, MR2004947}), 
number theory (cf. \cite{MR3726238, MR3773936, MR3773937, MR3773940}), 
combinatorics (cf. \cite{mansour2001gamma, MR1755751}) 
and mathematical physics (cf. \cite{MR3798008, MR3780488, MR3763251, MR3759662}). 
Furthermore, the beta function describes important properties of the strong nuclear force. That is, the nuclear interactions of elementary particles modeled using $1$-dimensional strings rather than using zero-dimensional particles are precisely described by the Euler beta function as a model for the scattering amplitude (cf. \cite{MR2006979, MR2284609}).

One of the important properties about the beta function is its close connection to the binomial coefficients that follow from Ramanujan's master theorem (cf. \cite{MR0106147} or Appendix):  
\begin{equation}\label{}
B (r, s)   =  
 \frac{\mathit{\Gamma} (r) \mathit{\Gamma} (s)}{\mathit{\Gamma} (r+s)}  
 =  
 \frac{\pi \sin (\pi(r + s))}{\sin (\pi r) \sin (\pi s)} \frac{\mathit{\Gamma} (1 - r - s)}{\mathit{\Gamma} (1 - r) \mathit{\Gamma} (1 - s)}  
 =  
 \frac{\pi \sin (\pi(r + s))}{\sin (\pi r) \sin (\pi s)} 
 { - r - s \choose - r }, \nonumber
\end{equation}
where the second equality holds only when $r,s\not\in \mathbb{Z}$, 
and the third equality holds in the sense of holomorphic functions, i.e., if the equality holds on a dense open set, then it holds everywhere, including at the singularities, due to the identity principle.

Let $G$ be a Lie group and let $P$  be a parabolic subgroup. 
In this article we introduce the Euler beta super-integral, generalizing the classical beta function:
\begin{equation}
B (r, s)   :=  
 \int_0^1 t^{r - 1} (1 - t)^{s - 1} dt , \hspace{4mm}\mbox{ where } 
 \text{Re}(r), \hspace{2mm}\text{Re}(s) >0, 
\end{equation}
where one integrates over a parameter cycle over a Grassmannian $G/P$. 

Our approach is based on the technique of Gelfand (cf. \cite{MR841131, MR888012, MR1127265, MR1121103}), who introduced a geometric interpretation for general hypergeometric functions as sections of tautological bundle $ \tau \to G_{k, n} $ over a complex Grassmannian manifold $ G_{k, n} $, i.e., see   \cite{MR1900449}. In particular, the classical beta function can be realized in terms of the geometry of the Grassmannian  
 $ G_{2, 3}$, which is explained in Section~\ref{subsection:connection-grassmannian-G23-beta}. 
 
 Although one could work in the category of complex differential geometry (and use a description of the fundamental bundles), we will work in the algebraic geometry
point of view as one can study Grassmannians in positive characteristic, which is very useful in arithmetic geometry and number theory. 



We consider the super-Grassmannian $G_{2|1,3|2}$, which consists of $2|1$-dimensional super-subspaces in $3|2$-dimensional complex superspace (cf. see Sections~\ref{subsection:super-grassmannians}, \ref{subsection:properties-super-Grass}, and \ref{section:classical-super-beta-integral}).  We then construct a certain $1$-form $\omega$ whose coefficient is a product of powers of linear forms on the super-Grassmann, and integrate it over a $1$-cycle $\gamma$, giving us the beta super-integral $ B (s, p_1, p_2; \xi, \xi', \eta) $ (also see Theorem~\ref{thm:Euler-super-integral-relation}). 
This integral that we construct in Section~\ref{section:classical-super-beta-integral} is one of the simplest possible generalizations of the classical Euler beta function, as one can add arbitrary many odd variables to the picture.

\begin{thrm}\label{thm:Euler-super-integral}
The super-integral of $\int_{\gamma}\omega$ on the super-Grassmannian $G_{2|1,3|2}$ is 
\begin{equation}
\Phi(s,p_1,p_2; \xi, \xi', \eta)  = - (-x_{21})^{-s-1} \int_0^1 \int_{-\infty}^{\infty} u^s (1 - u + \eta \theta)^{p_1} (\xi + \xi' u + \theta)^{p_2} \theta \, d \theta \,  d u .
\end{equation}
\end{thrm}

One may extend Theorem~\ref{thm:Euler-super-integral} to the super-monodromy setting using super-differential equations and super-connections. This construction will be discussed in our sequel paper. 


\begin{thrm}\label{thm:Euler-super-integral-relation}
The beta super-integral on the super-Grassmannian $G_{2|1,3|2}$ is 
\begin{equation}
\Phi(s,p_1,p_2; \xi, \xi', \eta)  = - (-x_{21})^{-s-1}  B (s, p_1, p_2; \xi, \xi', \eta). 
\end{equation}
\end{thrm}

We use the axiom of Berezin integration that $\int_{-\infty}^{\infty}\theta \: d\theta =1$ to obtain:  

\begin{lem}\label{lem:Euler-super-integral-classical-reduction}
If $ p_2 = 0 $ and 
$ \eta = 0 $ on $G_{2|1,3|2}$, 
we have 
\begin{equation}\label{eqn:reduction-classical-Euler-beta}
B (s, p_1, 0; \xi, \xi', 0)   =  
 \int_0^1 \int_{-\infty}^{\infty}  u^s (1 - u)^{p_1}   \theta \: d\theta d u 
 = 
  \int_0^1  u^s (1 - u)^{p_1}  d u 
    =   B (s + 1, p_1 + 1) ,
\end{equation}
with the right-hand side $ B $ being the classical Euler beta function.
\end{lem} 

\begin{thrm}\label{thrm:super-beta-gamma-function}
The beta super-integral on the super-Grassmannian $G_{2|1,3|2}$ also has the representation 
\begin{equation}
B (s, p_1, p_2; \xi, \xi', \eta) 
= 
\frac{\Gamma(s+1)\Gamma(p_1+1)}{\Gamma(p_1+s+2)}
\xi^{p_2}  \, 
{}_2F_1 \left( \begin{matrix} -p_2 , s+1 \\ p_1 + s+2 \end{matrix} \, \Big| - \frac{\xi'}{\xi}  \right). 
\end{equation}
\end{thrm}

\subsection{Summary of the sections}\label{subsection:summary} 
 
Section~\ref{section:background} recalls classical topics in the literature. Section~\ref{subsection:grassmannians} gives constructions of Grassmannians using quotients of algebraic groups, and we discuss integration on them. In Section~\ref{subsection:connection-grassmannian-G23-beta}, we give a connection between the geometry of a certain Grassmannian and the Euler beta integral. In Section~\ref{subsection:properties-beta-function-classical}, properties of the Euler beta function are given, including properties of the integral of a particular differential form over a cycle.  We define and construct vector superspaces in Section~\ref{subsection:superspace}, discussing superpoints and the Grassmann algebra. 
Super-Grassmannians are constructed and discussed in Section~\ref{subsection:super-grassmannians}, and in Section~\ref{subsection:properties-super-Grass}, we give basic properties of super-Grassmannians. Finally, in Section~\ref{subsection:change-of-variables}, we define change of variables in the super setting. 

In Section~\ref{section:classical-super-beta-integral}, we give the general construction of the beta super-integral (also see Section~\ref{subsection:general-beta-integral}), including a thorough exploration of the special case of the super-Grassmannian $G_{2|1,3|2}$ in Section~\ref{subsection:special-case}, thus proving Theorem~\ref{thm:Euler-super-integral} and Theorem~\ref{thm:Euler-super-integral-relation}. 
In Section~\ref{section:super-beta-integral-odd-variables}, we expand the beta super-integral with respect to the odd variables, thus proving Theorem~\ref{thrm:super-beta-gamma-function}. 
We conclude with an Appendix giving a more thorough discussion of Ramanujan's master theorem.

\subsection{Acknowledgement}\label{section:acknowledgement}

The first author is partially supported by National Research Council in Washington D.C. M.S.I. thanks Jan Kochanowski University in Kielce, Poland for their hospitality, where the initial draft of this paper was produced.  M.Z. would like to thank M.S.I. for introducing him to super-geometry and super-analysis. 
		
\section{Background}\label{section:background}

We begin with a discussion on Grassmannians. 
\subsection{Grassmannians}\label{subsection:grassmannians}
Let $ G_{k, n} $ denote the complex Grassmannian variety of $ k $-dimensional linear subspaces in an $ n $-dimensional  vector space. The atlas of $ G_{k, n} $ consist of maps $ x : G_{k, n} \to \mathbb{C}^{k (n - k) } $ constructed by the mappings from $ G_{k, n} $ to the space of matrices $ M_{k \times n} (\mathbb{C}) $, with rank equal to $ k $. They form open subsets in the space of matrices $ M_{k \times n} (\mathbb{C}) $ 
		as we exclude closed sets  
defined as the zero locus $ \det A_{i_1 i_2 \ldots i_k} = 0 $, where $ A_{i_1 i_2 \ldots i_k} $ is a minor whose columns have indices $ i_1 , i_2 , \ldots , i_k $.  
Since $G_{k,n}$ are isomorphic up to a change-of-basis, 
we will denote the open locus $\det A_{i_1 i_2 \ldots i_k} \not= 0 $ in $ M_{k \times n} (\mathbb{C})$  by 
 $ S_{k, n} $. 
The number of such (distinct) maps is equal to the number of choices of $ k $ columns from $ n $ columns, i.e., $ { n \choose k } $. 
 The choice of these columns associates a natural (left) action of $ GL_k ( \mathbb{C})$ on $ S_{k, n} $ defined by taking  the invertible minor to the identity matrix. 
Thus quotienting out $S_{k,n}$ by $ GL_k ( \mathbb{C})$, we obtain $ G_{k, n} = GL_k ( \mathbb{C}) \backslash S_{k, n} $. 
We also have the torus $T_k := GL_1 ( \mathbb{C})^k \cong GL_1(\mathbb{C})^{\oplus k}\oplus 1^{\oplus (n-k)}$  acting on $S_{k, n}$ 
diagonally from the right.  
Let us denote the possibly singular space of parameters  
 by $ P_{k, n}  = GL_k (\mathbb{C}) \backslash S_{k, n} / T_k $.  
 
For $1\leq i\leq k$, 
let $ l_i : = x_{1i} + x_{2i} u_2 + \ldots + x_{ki} u_k$,  
and let $ \phi_x (u) : = l_1^{s_1} l_2^{s_2} \cdots l_k^{s_k} $. 
We define hyperplanes as $ L_j = \{ l_j = 0 \}$, and 
divisors as $ D_x  = \bigcup_j L_j $.  
 
Let $ s \in U \subset \mathbb{C}^k $ denote a complex parameter. 
 Define
\begin{equation}\label{AomotoGelf}
\Phi (x, s)  :=  
 \int_{\gamma_x} \phi_x (u) \, du, 
\end{equation}
where  
$\gamma_x$ is a cycle on $ \mathbb{CP}^{k-1} \backslash (D_x \cap \mathbb{RP}^{k-1} $). 
Integrate over the cycle $ \gamma_x $, and homological with the cycle is a connected component of $ \mathbb{CP}^{k-1} \backslash (D_x \cap \mathbb{RP}^{k-1} $),  
lying in-between (and thus bounded by) the hyperplanes  $L_j$. Thus $\gamma_x$ is also bounded by $D_x$. 

The function $ \Phi $ is well-defined and it has an analytic continuation (along any curve) to the entire $ S_{k, n} $
by the identity principle.\footnote{The set of parameters is from $S_{k,n}$, and is invariant under the $GL_k(\mathbb{C})$-action, giving us $G_{k,n}$. Thus, we just integrate over $\mathbb{CP}^{k-1}$. 
} 
Furthermore, $ \Phi $ is \emph{invariant} with respect to the action of the special linear group $ SL_k ( \mathbb{C}) \subset GL_k (\mathbb{C}) $, which is evident in the following derivation. Define the left action by $GL_k (\mathbb{C})$ on $u$ by $g . u$, where $u=(u_1,\ldots, u_k)$.  
First, if $ g \in GL_k (\mathbb{C}) $, then
\begin{align}
\Phi (g. x, s) & = \,
 \int_{\gamma_{g . x}} \phi_{g . x} (u) \, du \, = \,
 \int_{\gamma_{g . x}} \phi_{x} (g^{-1} u) \, du \nonumber \\ & = \,
 (\det g)^{-1}  \int_{\gamma_{g . x}} \phi_{x} (v) \, dv . 
\end{align}
Since $ GL_k (\mathbb{C}) $ is a connected complex manifold, there exists a curve, say $ g (t) $, such that $ g (0) = I $ and $ g (1) = g $. It follows that  $ \gamma_{g . x} $ and $ \gamma_{x} $ are \emph{homologous}. 
But if $ g \in SL_k (\mathbb{C}) $, then 
\begin{equation}
\Phi (g. x, s) \, = \,
 \Phi (x, s) . 
\end{equation}
But $ GL_1 (\mathbb{C}) \to GL_k (\mathbb{C}) \to SL_k (\mathbb{C}) \overset{\det}{\longrightarrow} 1 $ is a \emph{principal bundle}, 
 and it follows that $ \Phi $ is defined on the quotient space $ SL_k (\mathbb{C}) \backslash S_{k, n} $, which is a line bundle over $ G_{k, n} $.

Let $ \tau \to G_{k, n} $ be the \emph{tautological bundle}, parametrized by all pairs of the form $ (x, V_x)$, where $x \in G_{k, n}$ and $ V_x \subset \mathbb{C}^n $ is a linear subspace corresponding to $ x \in G_{k, n} $. The rank of $ \tau $ is equal to $ k $, which implies that $ \bigwedge^k \tau $ is a line bundle. The action of $ SL_k (\mathbb{C}) $ on $ \tau $  leads to the natural action on $ \bigwedge^k \tau $ on each section of $ \tau $, invariant with respect to this action. 
This also produces an invariant section of $ \bigwedge^k \tau $. 
But since such section is precisely $ \Phi $, 
we will view $ \Phi $ as a \emph{section} of the line bundle $ \bigwedge^k \tau $.

\subsection{The geometry of the Grassmannian $G_{2,3}$ and the Euler beta function}\label{subsection:connection-grassmannian-G23-beta}

Let $ V_2 \subset \mathbb{C}^3 $ be a $2$-dimensional vector space. Then $ V_2 $ can be described using a basis of two linearly independent vectors from $ \mathbb{C}^3 $, which can be represented by the matrix
\begin{equation}\label{}
A  : = 
 \left( \begin{array}{ccc} x_{11} & x_{12} & x_{13} \\ x_{21} & x_{22} & x_{23} \end{array}
 \right)  
\end{equation}
of maximal rank (equal to two). Because the rank is maximal, we can find an invertible minor, say
\begin{equation}\label{}
g \, : = \,
 \left( \begin{array}{cc} x_{11} & x_{12} \\ x_{21} & x_{22}
\end{array}
 \right) , \quad g \in GL_2 (\mathbb{C}) , 
\end{equation}
such that  
\begin{equation}\label{}
g^{-1} A  : =  
 \left( \begin{array}{ccc} 1 & 0 & y_1 \\ 0 & 1 & y_2 \end{array}
 \right) , 
\end{equation}
with $ y_1, y_2 $ given explicitly in terms of $ x_{ij} $.

We define an affine map $ \phi (g^{-1} A) = (y_1, y_2) \in \mathbb{C}^2 $ on the \emph{Grassmannian variety} $ G_{2, 3} $, which parametrizes $ 2 $-dimensional vector subspaces in $ \mathbb{C}^3 $.

\begin{rem}
From the other two minors, we can construct other maps that cover the entire $ G_{2, 3} $, and transitions between these maps provide the gluing conditions.  That is, for 
\[ 
g' = 
\begin{pmatrix}
x_{11} & x_{13} \\ 
x_{21} & x_{23} \\ 
\end{pmatrix}\in GL_2(\mathbb{C}),  
\hspace{4mm} 
(g')^{-1}A = 
\begin{pmatrix}
1 & y_1' & 0 \\ 
0 & y_2' & 1 \\ 
\end{pmatrix}
\] 
while for 
\[
g''= 
\begin{pmatrix}
x_{12} & x_{13} \\ 
x_{22} & x_{23} \\ 
\end{pmatrix}\in GL_2(\mathbb{C}), 
\hspace{4mm} 
(g'')^{-1}A = 
\begin{pmatrix} 
y_1'' & 1 & 0  \\
 y_2'' & 0 & 1 \\  
\end{pmatrix}. 
\] 
\end{rem}

Now, let $ [w_1 : w_2] \in \mathbb{CP}^1 $. Then there are two maps of the form $ \phi :U\subseteq \mathbb{CP}^1 \to \mathbb{C} $, covering $ \mathbb{CP}^1 $. If we choose one of them, say $ w_1 \neq 0 $, then $ [w_1 : w_2] = [1 : t] $, with $ t \in \mathbb{C} $. From this data, we may construct three affine forms, i.e., 
the affine forms corresponding to $g^{-1}A$ are: 
\begin{equation}\label{}
 \left( \begin{array}{c} 1 \\ t \end{array}
 \right)^T
 \left( \begin{array}{ccc} 1 & 0 & y_1 \\ 0 & 1 & y_2 \end{array}
 \right)   =  
 \left( \begin{array}{ccc} 1 & t & y_1 + y_2 t \end{array}
 \right) . 
\end{equation}
From these forms, we construct a function
\begin{equation}\label{eqn:reduction-three-products}
\phi (t)   =  
 1 \cdot t^{\alpha} \cdot (y_1 + y_2 t)^{\beta} , 
\end{equation}
as well as the differential form $ \omega  = \phi (t) \, dt $.

To find the natural domain of integration for $ \omega $, note that the function $ \phi $ is multivalued and has singularities precisely when $ t = 0 $ or $ y_1 + y_2 t = 0 $. From the second equation, we deduce that the equality holds precisely when $ t = - y_1 / y_2 $. So $ \phi $, and thus also $ \omega $, is well-defined on the (universal cover of the) Riemann surface $ \mathbb{CP}^1 \backslash \{ 0, -y_1/y_2 , \infty \} $. Connecting these two points with a line gives the required 
 chain of integration, say $ \gamma $.\footnote{We say $ \gamma $ is a (twisted) \emph{cycle} in the \emph{twisted homology}. } To find the other two chains of integration, we consider the other map $U\subseteq \mathbb{CP}^1 \to \mathbb{C} $ that covers $ \infty $, and this point is also singular. We thus obtain three singular points  $ 0 , - y_1/ y_2 , \infty $, and connecting them pairwise gives the three 
 chains of integration.

\begin{defi}
Let $\phi$ be the map defined in \eqref{eqn:reduction-three-products}. 
Define the function 
\begin{equation}\label{eqn:one-cycle-one-form}
\Phi (\alpha, \beta) \, := \,
 \int_{\gamma} \omega . 
\end{equation} 
\end{defi}

\begin{defi}\label{defn:beta-integral}
For $\textrm{Re}(r), \textrm{Re}(s)>0$, 
the {\em Euler beta integral} is 
\begin{equation}\label{}
B (r, s) \, := \,
 \int_0^1 t^{r - 1} (1 - t)^{s - 1} \, dt . 
\end{equation}
\end{defi}

\begin{rem}\label{rem:Euler-beta-classical-monodromy}
The Euler beta integral in Definition~\ref{defn:beta-integral} may be extended and defined for any complex numbers $r$ and $s$  using analytic continuation of differentiation by parts under the integral, using a technique known as monodromy. 
\end{rem}

%
%

 

\subsection{Properties of the beta function}\label{subsection:properties-beta-function-classical}

We will describe several properties of $ \Phi $ in \eqref{eqn:one-cycle-one-form}, including its connection to the Euler beta integral. 

\begin{thrm}\label{phi=beta}
Up to the constant $(- 1)^{\alpha+1} y_1^{\beta + \alpha + 1} y_2^{-\alpha - 1}$, 
the function 
$ \Phi (\alpha, \beta) $ is equal to 
$ B (\alpha + 1, \beta + 1) $.
\end{thrm}

\begin{proof}
We have

\begin{align}\label{phibeta}
\Phi (\alpha, \beta) & = 
 \int_0^{- y_1/y_2} t^{\alpha} (y_1 + y_2 t)^{\beta} \, dt \nonumber \\
 & = 
 y_1^{\beta} \int_0^{- y_1/y_2} t^{\alpha} (1 + (y_2 / y_1) t)^{\beta} \, dt \nonumber \\
 & = 
- y_1^{\beta} \int_0^1 (- (y_1 / y_2) u)^{\alpha} (1 - u)^{\beta} \, (y_1 / y_2) du \nonumber \\ 
 & = 
 (- 1)^{\alpha+1} y_1^{\beta + \alpha + 1} y_2^{-\alpha - 1} \int_0^1 u^{\alpha} (1 - u)^{\beta} \, du  \nonumber \\
 & = 
 (- 1)^{\alpha+1} y_1^{\beta + \alpha + 1} y_2^{-\alpha - 1} B (\alpha + 1, \beta + 1) , 
\end{align}
where we use the substitution $u=-(y_2/y_1)t$. 
\end{proof}

\begin{rem}
The reduction 
of $ \Phi $ to $ B $ resembles properties of $ \mathbb{CP}^1 $ with its complex structure. In addition, one could guess the formula of Theorem \ref{phi=beta} knowing that the dimension of the moduli space of complex structures on the Riemann sphere $ \mathbb{CP}^1 $ is $0-3 = -3$ since $ \mathrm{dim} \, PGL_2 (\mathbb{C}) = 3 $.\footnote{More precisely, this means that the set of complex structures of $ \mathbb{CP}^1 $ is just a point, but $ \mathbb{CP}^1 $ also has a $ 3 $-dimensional automorphism group $ PGL_2 (\mathbb{C}) \simeq GL_2 (\mathbb{C}) / \mathbb{C}^* $ that allows for transformation of three arbitrary points to, say, $ \{ 0, 1, \infty \} $.} 




We give an alternative description, which is not as ``elegant''\footnote{This description is still elegant in some sense since the dimension of the moduli is always $ 3 g - 3 $, where $g$ is the genus. 
}
but it is easier to comprehend: the dimension of moduli space is zero, and it admits a free action by $3$-dimensional group of automorphisms. 
\end{rem}

\begin{rem}
Functions $ \Phi $ and $ B $ differ by a constant factor and thus are essentially the same, with the same qualitative properties. This can be understood in terms of a torus action that can be read in the proof of Theorem \ref{phi=beta}; in the second equality in the proof, we factor out the constant $y_1$   
and then by a change of variable for the integral, we factor out $y_2$.  
These procedures may be thought of as an adequate action by the torus $ \mathbb{C}^* \times \mathbb{C}^* $ on $ G_{2, 3} $.
\end{rem}

Permutation groups $ S_2 $ and $ S_3 $ act on the space $ G_{2, 3} $ by permuting the rows and the columns, respectively. The action of $ S_2 $ is not important as the basis elements for a $2$-dimensional subspace in a $3$-dimensional space produces the same $2$-dimensional space (regardless of the ordering on the basis vectors).
Furthermore, we can think of the $S_2$-action by permuting the rows of $G_{2,3}$ as restricting to the coordinate chart $\{[t:1] : t\in \mathbb{C} \}\subseteq \mathbb{CP}^1$, as opposed to the open set $\{ [1:t] : t\in \mathbb{C}\}$.\footnote{So the $S_2$-action can be thought of as moving between the charts, which is an additive operation.}

However, the action of $ S_3 $ is nontrivial. 
Since $\int_{\gamma}\omega = \int_0^1 l_1^{s_1}l_2^{s_2}l_3^{s_3}dt$, the permutation group $S_3$ acts on the exponents $s_i$, where the transposition $(i,i+1)$ swaps the exponents $s_i$ and $s_{i+1}$. 
Since $1^p = 1 $ for any $ p $, only the transposition $ (2, 3) $ acts on \eqref{eqn:reduction-three-products} in a nontrivial way.\footnote{There is a trivial factor $ 1^{\gamma} $ that is omitted, so the integrand reduces to $ 2 $ factors. Thus the biggest possible group that acts on the factors by permutation is $ S_2 $, and one checks by a straightforward computation that this is the subgroup of $ S_3 $ generated by the transposition $(2,3)$.}
This leads to the following result:

\begin{lem}
The action by the transposition $ (2, 3) $ on the space $ G_{2, 3} $ gives the formula
\begin{equation}\label{}
\Phi (\beta, \alpha) \, = \,
 ( - y_2)^{\alpha - \beta} \Phi (\alpha, \beta) .
\end{equation}
\end{lem}

\begin{proof}
We have
\begin{align}\label{}
\Phi (\beta, \alpha) & = \,
  (- 1)^{\beta+1} y_1^{\beta + \alpha + 1} y_2^{- \beta - 1} \int_0^1 u^{\beta} (1 - u)^{\alpha} \, du  \nonumber \\
 & = \,
  (- 1)^{\beta+1} y_1^{\beta + \alpha + 1} y_2^{- \beta - 1} \int_0^1 (1 - t)^{\beta} t^{\alpha} \, dt \nonumber \\
 & = \,
  (- 1)^{\beta - \alpha} y_2^{\alpha - \beta} \Phi (\alpha, \beta) \\ 
   & = \,
  (-y_2)^{\alpha - \beta} \Phi (\alpha, \beta) .
\end{align}
\end{proof}

Another property of $ \Phi $, as well as its more general counterpart, the hypergeometric function, is that it has close connection to a combinatorial function: the binomial coefficient
\begin{equation}
{ n \choose k } \, := \,
 \frac{n!}{(n - k)! k!} , 
\end{equation}
which can be shown in many ways, but we choose the one that can easily be generalized to other contexts. The main tool is \emph{Ramanujan's master theorem}, which we provide in more detail in the Appendix  (also see \cite{MR0106147}).

\begin{thrm}
The function $ \Phi $ is related to binomial coefficient by the equality
\begin{equation}\label{binom}
 \frac{\sin (\pi (\alpha+1))\sin (\pi(\beta+1))}{\pi \sin (\pi(\alpha+\beta+2))} 
 (- 1)^{\alpha+1} y_1^{- \beta - \alpha - 1} y_2^{\alpha + 1} \Phi (\alpha, \beta) \, = \,
 { - \alpha - \beta-2 \choose - \alpha-1  }.
\end{equation}
\end{thrm}
%

\textit{Proof}. \, Applying the homographic change of variables and transforming the unit interval to positive half-line, we get
\begin{align}\label{}
(- 1)^{\alpha+1} y_1^{- \beta - \alpha - 1} y_2^{\alpha + 1} \Phi (\alpha, \beta) & = \,
  \int_0^1 u^{\alpha} (1 - u)^{\beta} \, du  \nonumber \\
 & = \,
 \int_0^{\infty} t (1 + t)^{- \alpha - \beta-2} t^{\alpha - 1} dt. 
\end{align}
From the Taylor expansion of  $ (1 + t)^{- \alpha - \beta-2} $, 
we have 
\begin{equation}\label{eqn:power-series-expansion}
- t (1 + t)^{- \alpha - \beta-2}  =  
 \sum_{n \geq 0} { - \alpha - \beta-2 \choose n } (-t)^{n + 1} ,  
\end{equation}
and from Ramanujan's master theorem, we get (\ref{binom}). \, $ \square $ 


\subsection{Vector superspaces}\label{subsection:superspace}

Let $ (x_1, x_2, \ldots , x_n ; \xi_1, \ldots  , \xi_m) $ denote the standard variables on the superspace $ \mathbb{C}^{n|m} $; such a point is called a {\em superpoint}. 
The basis of differential forms is given by $ (dx_1, dx_2, \ldots  , dx_n ; d\xi_1, \ldots  , d\xi_m) \in T^* \mathbb{C}^{n|m} $. 
Since we have a natural $ n|m $-splitting on $ \mathbb{C}^{n|m} $, we first investigate an integral over an {\em odd point}, which involves only the odd coordinates, i.e., such a point lives on $\mathbb{C}^{0|m} $. 

The integral satisfies basic rules such as linearity or shift-invariance. As functions on $ \mathbb{C}^{0|m} $ are precisely elements of the Grassmann algebra $ \mathbb{C} [\xi_1, \ldots , \xi_m] $, we would also require that, at least, the basis elements are integrable.

Let $ \int : \mathbb{C} [\xi_1, \ldots, \xi_m] \to \mathbb{C} $ be a linear functional satisfying linearity and shift-invariance. Then the anti-commutation relations $ \xi_i \xi_j = - \xi_j \xi_i $ imply that 
$$ 
d\xi_i \wedge d\xi_j = -(-1)^{\bar{\xi_j}\bar{\xi_i}} d\xi_j \wedge d\xi_i = d\xi_j \wedge d\xi_i ,
$$ 
where 
$\bar{\xi_i}\in \mathbb{Z}_2$ is the parity of $\xi_i$. 
Thus
 $ \int d \xi_1 \wedge \cdots \wedge d\xi_m = 0$.\footnote{The translation invariance implies $ \int \varphi (\xi + \xi_0) d\xi = \int \varphi (\xi) d\xi $ for any constant $ \xi_0 $. 
So letting $ \varphi(\xi) = \xi $, we get $ \xi_0 \int d\xi = 0 $; it follows that $\int d\xi = 0 $.} 
But 
\begin{equation}\label{eqn:nonzero-odd-integral}
 \int \xi_m \cdots \xi_1 \, d \xi_1 \wedge \cdots \wedge d\xi_m \neq 0 ,
 \end{equation} 
 where the domain of integration does not need to be specified\footnote{It is a property of Berezin integral that the domain of integration is not specified, but rather, the integral should be algebraic and indefinite.}. 
 As there is a freedom of scaling, we can assume that the integral \eqref{eqn:nonzero-odd-integral} is equal to one.\footnote{This assumption is natural since not only does it mimics the unit area of the unit hypercube, but also from the point of view of dual basis as we can read the relation $ \int \xi_i \: d \xi_j = \delta_{ij} $ as the duality condition. Moreover, as an operation on $ \mathbb{C} [\xi] $, it follows that the integral operation $ \int \cdot \, d \xi $ is equivalent to differentiation.}
  Then for any (polynomial) function $ \varphi : \mathbb{C}^{0| m} \to \mathbb{C} $, we have
\begin{equation}\label{}
\int_{\mathbb{C}^{0|m}}
 \varphi (\xi) d \xi \, = \,
 \varphi_{1, 2, \ldots , m} ,
\end{equation}
where $ \varphi_{1, 2, \ldots , m} $ is the `top' coefficient in the Taylor super-expansion:
\begin{equation}\label{}
\varphi (\xi) \, = \,
 \varphi_0 + \sum_{i = 1}^m \varphi_i \xi_i + \sum_{i < j}^m \varphi_{i, j} \xi_i \xi_j + \ldots + \varphi_{1, 2, \ldots , m} \xi_1 \xi_2 \cdots \xi_m, 
\end{equation}
where the parity from the odd variables  has been incorporated in the coefficient of each monomial.

\subsection{Super-Grassmannians}\label{subsection:super-grassmannians} 

An extensive treatment on supermanifolds and super-Grassmannians is found in \cite{MR1172996,MR3588978,witten2012notes}. 

Let $ V^{k|l} \subset \mathbb{C}^{n|m} $ be a vector super-subspace of dimension $ k|l $. Then $ V^{k|l} $ can be described using a basis, i.e., a basis of $ k + l $ linearly independent $ \mathbb{Z}_2 $-graded vectors in $ \mathbb{C}^{n|m} $. But such basis can be put into a block matrix of the form  
\begin{equation}\label{super-basis}
A \, = \,
 \left( \begin{array}{cc} x & \xi \\ \eta & y \end{array} \right) 
\end{equation}
of maximal rank $ k + l $, 
where $ x \in \mathbb{C}^k \otimes \mathbb{C}^n $, $ \xi \in \mathbb{C}^k \otimes \mathbb{C}^m $, $ \eta \in \mathbb{C}^l \otimes \mathbb{C}^n $ and $ y \in \mathbb{C}^l \otimes \mathbb{C}^m $. If $ U $ is a minor of $ A $ of maximal rank $ k + l $, then the basis corresponding to $ A $ is equivalent to the basis corresponding to $ U^{-1} A$, and all bases of $ V^{k|l} $ can be constructed in this way (for different minors). As $ U \in GL_{k|l} (\mathbb{C}) $, we are led to the following: 

\begin{defi}[Super-Grassmannian]\label{defn:supergrassmannian}
Let $ B $ be the set of all bases of $ V^{k|l} $. The homogenous space $ G_{k|l , n|m} : = GL_{k|l} ( \mathbb{C})\backslash B $ is called the \emph{Grassmannian super-manifold}.
\end{defi}
 

We also refer to $G_{k|l , n|m}$ as the {\em super-Grassmannian}. 

Another construction of the super-Grassmannian is via the unitary supergroup $ U (n|m) $. The unitary superbasis of the subspace $ V_{k|l} $ can be described as the 
 unitary superbasis of $ \mathbb{C}^{n|m} $ modulo the action by $ U (n - k|m - l) $, which fixes $ n + m - k - l $ vectors, leaving the remaining $ k + l $ vectors arbitrary.

\begin{defi}[Stiefel super-manifold]
The homogenous space 
$$ S_{k|l , n|m} : = U (n|m) / U (n - k|m - l) 
$$ 
is called the \emph{Stiefel super-manifold}.
\end{defi}

In this super-unitary setting, the action of $ U (k|l) $ on the set of unitary basis of $ V^{k|l} $, i.e., on $ S_{k|l, n|m} $, corresponds to the action of $ GL_{k|l} (\mathbb{C}) $ on $ B $. Thus the definition of $ G_{k|l , n|m} $ is equivalent to the following: 

\begin{defi}\label{defn:Grassmannian-super-mfd}
The homogenous space 
$$ 
G_{k|l , n|m} : = U (n|m) / ( U (k|l) \times U (n - k|m - l)) 
$$ 
is called the \emph{Grassmannian super-manifold}.
\end{defi}

\begin{lem}
The two constructions of $ G_{k|l , n|m} $ in Definitions~\ref{defn:supergrassmannian} and \ref{defn:Grassmannian-super-mfd} are equivalent.
\end{lem}

\begin{proof}
The group $ U(n|m) $ is a deformation retract of $ GL_{n|m}(\mathbb{C}) $. This leads us to the same quotient space, with the use of a unitary vector in place of any  vector. 
\end{proof}


In Section~\ref{subsection:properties-super-Grass}, we will use both constructions to explore the geometry and topology of Grassmannian super-manifolds.

\subsection{Properties of the super-Grassmannian}\label{subsection:properties-super-Grass}

As $ G_{k|l , n|m} $ is an algebraic supervariety, it can be described with a notion of \emph{duality}, i.e., by the sheaf $ \mathcal{O}_{G_{k|l , n|m}} $ of regular functions\footnote{The reconstruction of $G_{k|l , n|m}$  
from $ \mathcal{O}_{G_{k|l , n|m}} $ is done with use of the theory of super-schemes, \textit{\`a la}  Grothendieck. This allows for further generalization, like inclusion of finite fields or, more generally, arbitrary super-commutative rings.}
 on $ G_{k|l , n|m} $. One can also construct an atlas which is 
 super-algebraic, super-analytic, and super-smooth.

Let $ A $ be the matrix given in (\ref{super-basis}). It is a super-matrix of maximal rank so there exists a $ k|l \times k|l $ minor, say $ U $, such that $ U \in GL_{k|l}( \mathbb{C}) $. Then the rows of $ U^{-1} A $ also form a basis of the same subspace $ V^{k|l} \subset \mathbb{C}^{n|m} $. In this way we have an equivalence relation: $ A $ and $ B $ are equivalent if there exists $ U \in GL_{k|l}( \mathbb{C}) $ such that $ B = U^{-1} A $, where the matrix $ U^{-1} A $ is of the form
\begin{equation}\label{redsuper-basis}
U^{-1} A  =  
 \left( \begin{array}{cccc} t & 1 & 0 & \theta \\ \upsilon  & 0 & 1 & u \end{array} \right),
\end{equation}
with  $ t \in \mathbb{C}^k \otimes \mathbb{C}^{n - k} $, $ \theta \in \mathbb{C}^k \otimes \mathbb{C}^{m - l} $, $ \upsilon \in \mathbb{C}^l \otimes \mathbb{C}^{n - k} $ and $ u \in \mathbb{C}^l \otimes \mathbb{C}^{m - l} $. As $ t, u, \theta, \upsilon $ are arbitrary, we obtain a map $ \varphi: G_{k|l , n|m} \to \mathbb{C}^{(n - k)k|(m - l)l} $. As there are 
 $ { n \choose k } { m \choose l } $ such maps, we can cover $ G_{k|l , n|m} $ with precisely $ { n \choose k } { m \choose l }$ super-charts.

To finish the construction of super-atlas on $ G_{k|l , n|m} $, we have to define the gluing condition on the intersections. Let $ \Omega $ be the chart associated to (\ref{redsuper-basis}) and let $ \Omega' $ be the chart associated to
\begin{equation}\label{}
(U')^{-1} A  =  
 B'  =  
 \left( \begin{array}{cccc} t' & 1 & 0 & \theta' \\ \upsilon' & 0 & 1 & u' \end{array} \right).
\end{equation}
Then the relation between the maps is given by 
$ B' =  (U')^{-1} A = (U')^{-1} U B $, corresponding to the gluing map $ \varphi' \circ \varphi^{-1}\in \Aut_{\mathbb{C}}(\mathbb{C}^{(n - k)k|(m - l)l})$.

\begin{exm}[The super-Grassmannian $ G_{2|2 , 4|4} $]
Let
\begin{equation}\label{}
A  = 
 \left( \begin{array}{cccccccc} x_{11} & x_{12} & x_{13} & x_{14} & \xi_{11} & \xi_{12} & \xi_{13} & \xi_{14} \\
x_{21} & x_{22} & x_{23} & x_{24} & \xi_{21} & \xi_{22} & \xi_{23} & \xi_{24} \\ \eta_{11} & \eta_{12} & \eta_{13} & \eta_{14} & y_{11} & y_{12} & y_{13} & y_{14} \\ \eta_{21} & \eta_{22} & \eta_{23} & \eta_{24} & y_{21} & y_{22} & y_{23} & y_{24} \end{array} \right).
\end{equation}
Then $ U^{-1} A = B $ is of the form
\begin{equation}\label{}
B   =  
 \left( \begin{array}{cccccccc} t_{11} & t_{12} & 1 & 0 & 0 & 0 & \theta_{13} & \theta_{14} \\
t_{21} & t_{22} & 0 & 1 & 0 & 0 & \theta_{23} & \theta_{24} \\ \upsilon_{11} & \upsilon_{12} & 0 & 0 & 1 & 0 & u_{13} & u_{14} \\ \upsilon_{21} & \upsilon_{22} & 0 & 0 & 0 & 1 & u_{23} & u_{24} \end{array} \right),  
\end{equation}
and $ (U')^{-1} A = B' $ is of the form
\begin{equation}\label{}
B'  =  
 \left( \begin{array}{cccccccc} t_{11} & 1 & t_{13} & 0 & 0 & 0 & \theta_{13} & \theta_{14} \\
t_{21} & 0 & t_{23} & 1 & 0 & 0 & \theta_{23} & \theta_{24} \\ \upsilon_{11} & 0 & \upsilon_{13} & 0 & 1 & 0 & u_{13} & u_{14} \\ \upsilon_{21} & 0 & \upsilon_{23} & 0 & 0 & 1 & u_{23} & u_{24} \end{array} \right).
\end{equation}
The matrix $ U $ is formed from $ A $ by choosing columns $3$, $4$, $5$, and $ 6 $ and $ U' $ is formed from $ A $ by choosing columns $2$, $4$, 
$5$, and $ 6 $. They are, respectively, of the form
\begin{equation}\label{}
 \left( 
 \begin{array}{cccc} 
 x_{13} & x_{14} & \xi_{11} & \xi_{12} \\
 x_{23} & x_{24} & \xi_{21} & \xi_{22} \\ 
 \eta_{13} & \eta_{14} & y_{11} & y_{12} \\ 
 \eta_{23} & \eta_{24} & y_{21} & y_{22} 
 \end{array} 
 \right) 
\hspace{4mm}
\mbox{ and }
\hspace{4mm}
 \left( 
 \begin{array}{cccc} 
 x_{12} & x_{14} & \xi_{11} & \xi_{12} \\
 x_{22} & x_{24} & \xi_{21} & \xi_{22} \\ 
 \eta_{12} & \eta_{14} & y_{11} & y_{12} \\ 
 \eta_{22} & \eta_{24} & y_{21} & y_{22} 
 \end{array} 
 \right).
\end{equation}

On the intersection of the maps 
$\Omega $ and $ \Omega ' $, 
both $ U $ and $ U' $ are invertible. Thus $ U^{-1} $ and $ (U')^{-1} $ 
 are well-defined. This implies that the product $ (U')^{-1} U $ is also well-defined, 
 providing a transition map between the charts corresponding to $ B $ and $ B'$, which is $(U')^{*} U$.\footnote{It does not matter whether we take a general or the unitary point of view, but the computations are simpler in the unitary setting.} 
\end{exm} 

\subsection{Change of super-variables}\label{subsection:change-of-variables}

Having the odd part of the integral defined, we introduce the full integral of a (compactly supported) function on the superspace $ \mathbb{C}^{n|m} $. The function $ \varphi : \mathbb{C}^{n|m} \to \mathbb{C} $ can be expanded into a power series in the odd variables as
\begin{equation}\label{}
\varphi (x, \xi) \, = \,
 \varphi_0 (x) + \sum_{i = 1}^m \varphi_1 (x) \xi_i 
 + \ldots + \varphi_{1, 2, \ldots , m} (x) \xi_1 \xi_2 \cdots \xi_m , \nonumber
\end{equation}
where, this time, each coefficient $\varphi_{\alpha}(x)$ is an element of the set $ C_c^{\infty} (\mathbb{C}^n)$ of compactly-supported smooth functions on $\mathbb{C}^n$. This implies that the integral $ \int \varphi \: dx  d\xi $ is of the form
\begin{equation}\label{eqn:power-series-odd-variables}
\int_{\mathbb{C}^{n|m}} \varphi (x, \xi) dx d\xi  =  
 \int_{\mathbb{C}^n} \varphi_{1, 2, \ldots , m} (x) \, dx  
\end{equation}
by using properties of Berezin integral, i.e., $ \int_{\mathbb{C}^{0|1}} d\xi = 0 $ and $ \int_{\mathbb{C}^{0|1}} \xi d\xi = 1 $.  
We thus have a super-analogue of Fubini's theorem: 
\begin{equation}
\int_{\mathbb{C}^{n|m}}  \varphi(x,\xi) dx d\xi = 
\int_{\mathbb{C}^n} \varphi_{1, 2, \ldots , m} (x) dx \int_{\mathbb{C}^{0|m}}\xi_1\xi_2\cdots \xi_m d\xi.  
\end{equation} 

Now, the following analysis is an important generalization of classical concepts in analysis, geometry and topology. To proceed, we first deliver a formula for the change of super-variables, paving way to prove the inverse and implicit function theorems, leading us to an analysis on supermanifolds.

Let $ (y, \eta) = f (x, \xi) $. The function $ f $ has the derivative of the form\footnote{The partial derivatives $ \partial y / \partial x $, $\partial y/\partial \xi$, $\partial \eta/\partial x$, and $\partial \eta/\partial \xi$ denote the blocks of super-derivative.}
\begin{equation}\label{}
D f  =  
 \left( \begin{array}{cc} \dfrac{\partial y}{\partial x} & \dfrac{\partial y}{\partial \xi} \\ \dfrac{\partial \eta}{\partial x} & \dfrac{\partial \eta}{\partial \xi} \end{array} \right).
\end{equation}
The Jacobian, however, cannot be computed in the usual way. Instead, one uses the super-determinant, or {\em Berezinian}, defined as
\begin{equation}\label{}
\mathrm{sdet} \, D f  =  
 \det \left( \dfrac{\partial y}{\partial x} - \dfrac{\partial y}{\partial \xi} \left( \dfrac{\partial \eta}{\partial \xi}\right)^{-1} \dfrac{\partial \eta}{\partial x} \right) \det \left( \dfrac{\partial \eta}{\partial \xi}\right)^{-1} , \nonumber
\end{equation}
where $ \partial \eta / \partial \xi $ must be invertible. 
Thus we have the change of super-variables, and the corresponding integrals 
\begin{equation}\label{} 
\int \varphi (y, \eta) \: d y d \eta 
\end{equation}
and
\begin{equation}\label{}
\int \varphi ( f (x, \xi))
\det \left( \dfrac{\partial y}{\partial x} - \dfrac{\partial y}{\partial \xi} \left( \dfrac{\partial \eta}{\partial \xi}\right)^{-1} \dfrac{\partial \eta}{\partial x} \right) \det \left( \dfrac{\partial \eta}{\partial \xi}\right)^{-1}  
d x d \xi \nonumber
\end{equation}
are equal.\footnote{We do not need to insert an absolute value around the two determinants in the integral since the Berezinian is an even function.} 

\begin{rem}
The change of super-variables is implicitly and minimally used in the construction of beta super-integral (see the proof of Theorem~\ref{thm:Euler-super-integral} in Section~\ref{subsection:special-case}). 
\end{rem} 

\section{Beta super-integral}\label{section:classical-super-beta-integral} 

For $l\leq m$,  
let $ A' $ be the super-matrix of the form
\begin{equation}
\left( 
\begin{matrix} 
x' & \xi' \\ 
\eta' & y'
\end{matrix}
 \right),
\end{equation}
with $ x' \in \mathbb{C}^2 \otimes \mathbb{C}^3 $, $ \xi' \in \mathbb{C}^2 \otimes \mathbb{C}^m $, $ \eta' \in \mathbb{C}^l \otimes \mathbb{C}^3 $ and $ y' \in \mathbb{C}^l \otimes \mathbb{C}^m $ such that 
$ A' $ has maximal rank. 
Then we can regard $ A' $ as a supersymmetric extension of a $ 2 \times 3 $ matrix $ x'$, of maximal rank. We can also identify $A'$ with the set of $2+l$ vectors in the vector superspace $ \mathbb{C}^{3|m}$.
 Because of this, one can regard $ A' $ as a point of the Grassmannian manifold $ G_{2|l, 3|m} $. Assuming that the invertible $ 2|l \times 2|l $ minor $ U $ is constructed from the columns 
$2$, $3$, $4, \ldots$, $l + 3 $ (see Section~\ref{subsection:properties-super-Grass} for more detail on our choice for these columns), 
we can construct a map\footnote{This map is actually a super-biholomorphic isomorphism $ \varphi : G_{k|l , n|m} \to \mathbb{C}^{k|l \times  (n - k|m - l)} $ since $G_{k|l , n|m}$ is locally isomorphic to $\mathbb{C}^{k|l \times  (n - k|m - l)}$. 
But since  $ \mathbb{C}^{k|l \times  (n - k|m - l)} $ and $ \mathrm{Hom} (\mathbb{C}^{k|l}, \mathbb{C}^{n - k | m - l}) $ are the same as vector superspaces, we will identify the two.}
 $ A = U^{-1} A' $, which is of the form  
\begin{equation}\label{eqn:matrixA-G-23-LM}
A  =  
 \left( \begin{matrix} x & I_2 & 0 & \xi \\ \eta & 0 & I_l & y
\end{matrix}
 \right) ,
\end{equation}
with $ x \in \mathbb{C}^2 \otimes \mathbb{C} $, $ \xi \in \mathbb{C}^2 \otimes \mathbb{C}^{m - l} $, $ \eta \in \mathbb{C}^l \otimes \mathbb{C} $ and $ y \in \mathbb{C}^l \otimes \mathbb{C}^{m - l} $, 
with all nonconstant entries being arbitrary.

To $ x $ one can (in the classical setting) associate the function 
\begin{equation}
\Phi (s, x) \, := \,
 \int_{\gamma_x} \omega_x ,
\end{equation}
where $ \omega_x = 
 l_1^{s_0} l_2^{s_1} l_3^{s_2} \, dt $, 
 $l_1 = x_{11}+x_{21}t$, $l_2 = 1$, $l_3 = t$, 
 and 
 $ \gamma_x $ is a cycle associated to a $ 3 $-point arrangement in $ \mathbb{CP}^1 $, bounded by the equations $  l_i = 0 $. 

In what follows, we will describe the super-generalization of this construction.

\subsection{Construction of the general beta super-integral}\label{subsection:general-beta-integral}
 
 The projective superspace $\mathbb{CP}^{1|l}$ has homogeneous coordinates $[ t'_1 : t'_2 
 | \theta'_1 , \ldots , \theta'_l ] $.  
  If $t_1'\not=0$, then $[ t'_1 : t'_2 
 | \theta'_1 , \ldots , \theta'_l ] = [1 : t 
 | \theta_1 , \ldots , \theta_l ]$, which can be viewed as 
the point 
\begin{equation}
( 1 : t 
 | \theta_1 , \ldots , \theta_l ) \in \mathbb{C}^{1|l} .
\end{equation}

The reduced matrix $ A $ in \eqref{eqn:matrixA-G-23-LM} acts on $ \mathbb{CP}^{1|l} $ from the right as multiplication (cf. \cite{MR3588978}): 
\begin{equation}\label{eqn:matrixA-G-23-LM-multn}
[ t'_1 : t'_2 
 | \theta'_1 , \ldots , \theta'_l ]  
 \left( \begin{matrix} x & I_2 & 0 & \xi \\ \eta & 0 & I_l & y
\end{matrix}
 \right).  
\end{equation}
Thus it also acts on the space $ \mathbb{C}^{1|l} $ from the right, and the result is a family of affine forms: three even\footnote{By even, we mean that the variables 
$ \xi : = \eta_{ij} \theta_i $ and $ \xi' : = \eta_{ij}' \theta_i' $ satisfy 
$ \xi ' \xi = \eta_{ij}' \theta_i' \eta_{ij} \theta_i = \eta_{ij} \eta_{ij}' \theta_i' \theta_i = \eta_{ij} \theta_i \eta_{ij}' \theta_i' = \xi \xi' $, so the minus sign does not apply in their case when it comes to the commutation relations. Thus they are even.}
  forms:  
\begin{align}
l_1 &  =  x_{11} + x_{21} t + \eta_{11} \theta_1 + \eta_{21} \theta_2+ \ldots + \eta_{l1} \theta_l,   \nonumber \\
l_2 &  =  1,  
\nonumber \\
l_3 &  = t,    
 \nonumber 
\end{align} 
and $ m $   odd forms: 
\begin{align}
\lambda_i & =   \theta_i  \hspace{4mm}\mbox{ for } 1 \leq  i \leq l  ,   \nonumber \\
 \lambda_{l + 1} &  =  
 \xi_{1, l + 1} + \xi_{2, l + 1} t + y_{1, l + 1} \theta_1 + \ldots + y_{l, l + 1} \theta_l ,  \nonumber \\
  \lambda_{l + 2} &  =  
 \xi_{1, l + 2} + \xi_{2, l + 2} t + y_{1, l + 2} \theta_1 + \ldots + y_{l, l + 2} \theta_l,  \nonumber \\
 &\hspace{2mm} \vdots \, \nonumber \\
 \lambda_m &  =  
 \xi_{1, m} + \xi_{2, m} t + y_{1, m} \theta_1 + \ldots + y_{l, m} \theta_l .\nonumber
\end{align}

 \begin{defi}\label{defn:super-beta-integral-general}
The {\em super-beta integral} is defined as 
\begin{equation}\label{sup-bet}
\Phi (s|\sigma ; x, y, \xi, \eta) \, : = \,
  \int_{\gamma_x \times \mathbb{C}^{0|l}} l_1^{s_1} l_2^{s_2} l_3^{s_3} \lambda_1^{\sigma_1} \cdots \lambda_m^{\sigma_m} dt \wedge d \theta_1 \wedge \cdots \wedge d \theta_l .
\end{equation}
\end{defi}

%

\subsection{Special case of the beta super-integral}\label{subsection:special-case}
We will first prove Theorem~\ref{thm:Euler-super-integral}.

\begin{proof}
Consider the matrix
\begin{equation}
A  =  
 \left(
 \begin{array}{ccccc}
 x_{11} & x_{12} & x_{13} & \xi_{11} & \xi_{12} \\
 x_{21} & x_{22} & x_{23} & \xi_{21} & \xi_{22} \\
 \eta_{11} & \eta_{12} & \eta_{13} & y_{11} & y_{12}
 \end{array} \right),
\end{equation}
of maximal rank. This matrix can be identified with the set of $ 3 $ vectors in the vector superspace $ \mathbb{C}^{3|2} $. Furthermore, as the rank is maximal, the image of the map that corresponds to $ A $ is a $ 2 | 1 $-dimensional super-subspace since the rows correspond to linearly independent vectors, i.e., $2$ even vectors and $1$ odd vector.

As $ A $ is of maximal rank, there exists a non-zero minor, say $ U $. Without loss of generality, let's assume that $ U $ is of the form
\begin{equation}
U  =  
 \left(
 \begin{array}{ccc}
 x_{12} & x_{13} & \xi_{11} \\
 x_{22} & x_{23} & \xi_{21} \\
 \eta_{12} & \eta_{13} & y_{11} 
 \end{array} \right).
\end{equation}
Then the $ 2|1 $-subspace of $ \mathbb{C}^{3|2} $ that is described by $ A $ is also described by  $U^{-1}A$, which is of the form
\begin{equation}
U^{-1} A   =  
 \left( 
 \begin{array}{ccccc}
 x'_{11} & 1 & 0 & 0 & \xi'_{12} \\
 x'_{21} & 0 & 1 & 0 & \xi'_{22} \\
 \eta'_{11} & 0 & 0 & 1 & y'_{12}
 \end{array} \right).
\end{equation}
With the data provided by $ U^{-1}A $, we can form a family of two linear forms, $ x'_{11} t_1 + x'_{21} t_2 + \eta'_{11} \theta_1 $ and $ \xi'_{12} t_1 + \xi'_{22} t_2 + y_{12}' \theta_1 $, together with three characters: $2$ even ones $ t_1, t_2 $ and $1$ odd one $ \theta_1 $, by multiplying on the left by $(t_1,t_2;\theta_1)$.

The affine counterparts are 
\begin{equation}\label{eqn:affine-counterparts}
\begin{split}
l_1 & =  1 + x_{21}' t + \eta_{11}' \theta,  \\
l_2 & =  \xi_{12}' + \xi_{22}' t + \theta ,  \\ 
\end{split}
\end{equation}
where $x_{11}'$ and $y_{12}'$ in \eqref{eqn:affine-counterparts} are canceled under the change of variables.  

For the ease of reading, let us omit the primes over the coefficients. 
The group, which depends locally,   
acts on the set of linear forms $l_i$, which, in fact, acts on the powers $l_i^{\alpha_i}$ of the linear forms. 
After fixing exponents, one constructs the form
\begin{equation}
\omega \, := \,
 t^s l_1^{p_1} l_2^{p_2} \theta \: d \theta \: d t ,
\end{equation}
which can be integrated (for certain values of $ s, p_1, p_2 $) over a region $ \gamma $ in $ \mathbb{C}^{1|1}=\Spec(\mathbb{C}[t;\theta])$.
 The first factor constrains us to $ t > 0 $, while from $ l_1 $ it follows that 
 $ t > - x_{21}^{-1} $ 
since if in, say, $ (1 - tx)^{\alpha} $, we allow both $ t x < 0 $ and $ t x > 0 $, 
then there is a jump coming from the multi-valuedness of $ (1 - tx)^{\alpha} $.\footnote{In the classical case of logarithms, choosing a loop around the singularity and running around brings one to another point other than the original point (the new point will differ by $\pm 2 \pi $, depending on the orientation of the loop).} 
This can be overcome by considering the universal covering for complex spaces and their monodromy, or, so called, twisted cohomology (but in this manuscript, we will rule out these cases by restricting to the parameters of the classical beta function). 
  
%

The parameter $ \theta $ is allowed to be arbitrary. Thus, we get the following formula:  
\begin{equation}
\int_{\gamma} \omega \, := \,
 \int_{- 1/x_{21}}^0 \int_{-\infty}^{\infty}  
 t^s (1 + x_{21} t + \eta_{11} \theta)^{p_1} (\xi_{12} + \xi_{22} t + \theta)^{p_2} \theta \: d \theta d t .
\end{equation}
This can be further simplified using the change of variables 
$ x_{21} t = - u $, with $ dt = - du / x_{21} $: 
\begin{equation}
\int_{\gamma} \omega  = 
- (- x_{21})^{-s - 1} \int_0^1 \int_{-\infty}^{\infty} u^s (1 - u + \eta_{11} \theta)^{p_1} (\xi_{12} - \xi_{22} x_{21}^{-1} u + \theta)^{p_2} \theta \: d \theta d u .  
\end{equation}
Further simplification is done by letting $\eta = \eta_{11}$, $\xi = \xi_{12}$ and  
$\xi'= - \xi_{22} x_{21}^{-1} $ as only the integral is of importance. %
This concludes the proof.
\end{proof}  

\begin{rem}\label{rem:Theorem1-monodromy}
We will discuss in our future work the case when $t x_{21}+1 \leq 0$ for Theorem~\ref{thm:Euler-super-integral} by considering monodromy for the super-variables. 
\end{rem}

We will now prove Theorem~\ref{thm:Euler-super-integral-relation}.
\begin{proof}
By Theorem~\ref{thm:Euler-super-integral}, we have 
\begin{equation}
\Phi(s,p_1,p_2; \xi, \xi', \eta)  = - (-x_{21})^{-s-1} \int_0^1 \int_{-\infty}^{\infty} u^s (1 - u + \eta \theta)^{p_1} (\xi + \xi' u + \theta)^{p_2} \theta \: d \theta d u .
\end{equation}
Since $B(s,p_1,p_2; \xi, \xi', \eta)$ differs from $\Phi(s,p_1,p_2; \xi, \xi', \eta)$ by a scalar, 
we conclude that 
$$
\Phi(s,p_1,p_2; \xi, \xi', \eta)  = - (-x_{21})^{-s-1} B(s,p_1,p_2; \xi, \xi', \eta), 
$$
where 
\[ 
B(s,p_1,p_2; \xi, \xi', \eta) = \int_0^1 \int_{-\infty}^{\infty} u^s (1 - u + \eta \theta)^{p_1} (\xi + \xi' u + \theta)^{p_2} \theta \: d \theta d u .
\] 
\end{proof}

\section{Integrating the super-beta integral with respect to the odd variables}\label{section:super-beta-integral-odd-variables}

We will finally prove Theorem~\ref{thrm:super-beta-gamma-function}. 

\begin{proof}
Recall the binomial expansion 
\begin{equation}
(x + y)^p \, = \,
 \sum_{k \geq 0} { p \choose k } x^{p - k} y^k .
\end{equation}
Applying it to powers of the affine forms under the integral in Theorem~\ref{thm:Euler-super-integral}, 
 we get
\begin{equation}\label{eqn:affine-forms-power-series-expansion}
\begin{split}
(1 - u + \eta \theta)^{p_1} & =  
 \sum_{k \geq 0} { p_1 \choose k } (1 - u)^{p_1 - k} (\eta \theta)^k,  \\
(\xi + \xi' u + \theta)^{p_2} & =  
 \sum_{k \geq 0} { p_2 \choose k } (\xi + \xi' u)^{p_2 - k} \theta^k , \\ 
 \end{split} 
\end{equation}
but $ \theta^2 = 0 $ and $ \eta^2 = 0 $ 
by supersymmetry. That is, 
$ \xi_i \xi_j = - \xi_j \xi_i $ for any super-variable $ \xi_i $, 
so $ \xi_i^2 = - \xi_i^2 $. 
This implies that $ \xi_i^2 = 0 $.  
It follows that $ (\eta \theta)^2 = 0 $. 
The equalities in \eqref{eqn:affine-forms-power-series-expansion} simplify to
\begin{equation}\label{eqn:eqn:affine-forms-poly}
\begin{split}
(1 - u + \eta \theta)^{p_1} & =  
 (1 - u)^{p_1} + p_1 (1 - u)^{p_1 - 1} \eta \theta ,  \\
(\xi + \xi' u + \theta)^{p_2} & =  
 (\xi + \xi' u)^{p_2} + p_2 (\xi + \xi' u)^{p_2 - 1} \theta . \\ 
 \end{split}
\end{equation}
Multiplying the two equalities in \eqref{eqn:eqn:affine-forms-poly} together, we get\footnote{Variables $ \xi $ and $ \xi' $ are odd; however, the parity operator is not relevant as the only summand that contributes to the integral is the one with $ 0 $-th power of $ \theta $.}
\begin{equation}\label{eqn:multi-two-eqns}
\begin{split}
(1 - u + \eta &\theta)^{p_1} (\xi + \xi' u + \theta)^{p_2} 
=  
 (1 - u)^{p_1} (\xi + \xi' u)^{p_2}  \\ 
 &\hspace{4mm}  
 +  \cancel{p_1 (1 - u)^{p_1 - 1} \eta \theta  (\xi + \xi' u)^{p_2}}  
 +  \cancel{(1 - u)^{p_1}  p_2 (\xi + \xi' u)^{p_2 - 1} \theta } 
 + \ldots 
 \end{split}
\end{equation}
where the two terms cancel due to supersymmetry. 
Since $\int_{-\infty}^{\infty} \theta \, d\theta =1$, 
the beta super-integral for $ B $ then reduces to 
\begin{equation}\label{eqn:integral-formula-for-B}
\begin{split}
B (s, p_1, p_2; \xi, \xi', \eta)  &=  
 \int_0^1 u^s (1 - u)^{p_1} (\xi + \xi' u)^{p_2} \, d u  \\
 &=  
 \xi^{p_2} \int_0^1 u^s (1 - u)^{p_1} \left(1-\left( - \frac{\xi'}{\xi}\right) u\right)^{p_2} \, du. 
 \end{split}
\end{equation}
Since 
\[ 
\beta(s+1,p_1+1) = \frac{\Gamma(s+1)\Gamma(p_1+1)}{\Gamma(p_1+s+2)} 
\hspace{4mm} 
\mbox{ for } s+1,p_1+1 >0 
\] 
(see 
Theorem 7 on page 19 in \cite{MR0107725} or 
Theorem 2.5.9 on page 46 in \cite{hannah2013identities}), 
and since Euler-Gauss hypergeometric function (see Theorem 3.4.1 on page 71 in \cite{hannah2013identities}) gives us 
\[ 
_2F_1 \left( \begin{matrix} -p_2 , s+1 \\ p_1 + s+2 \end{matrix} \, \Big| - \frac{\xi'}{\xi}  \right)
=  \frac{\Gamma(p_1+s+2)}{\Gamma(s+1)\Gamma(p_1+1)} 
\int_0^1 u^s (1-u)^{p_1}  \left(1-\left( - \frac{\xi'}{\xi}\right) u\right)^{p_2} \, du, 
\] 
we have 
\[ 
\xi^{p_2}  \, 
{}_2F_1 \left( \begin{matrix} -p_2 , s+1 \\ p_1 + s+2 \end{matrix} \, \Big| - \frac{\xi'}{\xi}  \right)
= 
\frac{\Gamma(p_1+s+2)}{\Gamma(s+1)\Gamma(p_1+1)} B (s, p_1, p_2; \xi, \xi', \eta). 
\] 
We conclude 
\[ 
B (s, p_1, p_2; \xi, \xi', \eta) 
= 
\frac{\Gamma(s+1)\Gamma(p_1+1)}{\Gamma(p_1+s+2)}
\xi^{p_2}  \, 
{}_2F_1 \left( \begin{matrix} -p_2 , s+1 \\ p_1 + s+2 \end{matrix} \, \Big| - \frac{\xi'}{\xi}  \right). 
\] 
\end{proof}


%
%
%






%




\section*{Appendix}

We will now state and prove Ramanujan's master theorem for completeness. Also see \cite{MR2994092}, after Theorem 3.2 on page 4.

\begin{thrm}\label{thm:sinphi}
If 
\begin{equation}
\frac{\pi}{\sin \pi s} F (- s) \, = \,
 \int_0^{\infty} t^{s - 1} \varphi (t) \, dt ,
\end{equation}
then
\begin{equation}
\varphi (t) \, = \,
 \sum_{n \geq 0} (- t)^n F (n).
\end{equation}
\end{thrm}

\begin{proof}
If
\begin{equation}
F (s) \, = \,
 \int_0^{\infty} t^{s - 1} f (t) \, dt 
 \quad
 \mbox{ and }
 \quad
G (s) \, = \,
 \int_0^{\infty} t^{s - 1} g (t) \, dt , \nonumber
\end{equation}
then by Mellin convolution, we have
\begin{equation}
F (s) G (s) \, = \,
 \int_0^{\infty} t^{s - 1} h (t) \, dt ,
\end{equation}
where 
\begin{equation}
h (t) \, = \,
 \int_0^{\infty} \frac{f (u) g (t / u)}{u} \, du .
\end{equation}

In particular, 
\begin{equation}\label{eqn:Ramanujan-proof-end}
\frac{\pi}{\sin \pi s} F (- s) \, = \,
 \int_0^{\infty} t^{s - 1} \varphi (t) \, dt ,  \quad \mbox{ where }
 \varphi (t) \, = \,
 \int_0^{\infty} \frac{f (1/u)}{u + t} \, du .
\end{equation}
We have $ F (- s) $ in place of $ F (s) $, which translates from $  f (u) $ to $ f (1/u) $: 
\[ 
F(-s) = \int_0^{\infty} u^{s-1} f(1/u)\, du.  
\]  
Also, 
$$ G (s) = \frac{\pi}{\sin \pi s} = \int_0^{\infty} \frac{t^{s - 1}}{1 + t}\, dt ,
$$
so the convolution of $F(-s)$ with $ \pi / \sin \pi s $ is
\begin{equation}\label{}
\frac{\pi}{\sin \pi s} F (- s) \, = \,
 \int_0^{\infty} t^{s - 1} \int_0^{\infty} \frac{f (1/u)}{1 + t/u} \, \frac{du}{u} \, dt = \,
 \int_0^{\infty} t^{s - 1} \int_0^{\infty} \frac{f (1/u)}{u + t} \, du 
 \, dt.
\end{equation}

Replacing $ (u + t)^{-1} $ by its Taylor series, we obtain 
\begin{align*}
 \int_0^{\infty} \frac{f (1/u)}{u + t} \, du 
 &=  \int_0^{\infty}  \frac{u^{-1} f (1/u)}{1 + (t/u)} \, du \\
 &=  \int_0^{\infty} \frac{f(1/u)}{u} \sum_{n\geq 0} (-t/u)^n du \\ 
 &= \sum_{n\geq 0}  (-t)^n \int_0^{\infty} u^{-n-1} f(1/u)\, du.  
\end{align*}
Now, letting $v=1/u$ (and so $dv = -u^{-2}du$),  we see that 
\begin{align*} 
\int_0^{\infty} u^{-n-1} f(1/u)\, du &= \int_0^{\infty} v^{n+1}f(v)\frac{dv}{v^2}  
= \int_0^{\infty} v^{n-1}f(v)\, dv  
= F(n). 
\end{align*} 
\end{proof}

\bibliography{super-radon}


\end{document}